\documentclass[a4paper,11pt,pointlessnumbers]{scrartcl}

\usepackage[utf8]{inputenc}
\usepackage{amsthm}
\usepackage{amssymb}
\usepackage{enumerate}

\usepackage{fullpage}

\usepackage{lmodern}
\usepackage{MnSymbol}
\usepackage{booktabs}
\usepackage{tabularx}

\usepackage{hyperref}
\usepackage{cleveref}

\usepackage{tikz}
\usetikzlibrary{hobby}
\usepackage{tikz-cd}
\usepackage{bm}

\newcommand{\problem}[6][ ]{
\begin{prob}
Compute a function $#2$ with the following properties:\\
\begin{tabular}{@{}p{0.125\linewidth}@{}p{0.875\linewidth}@{}}
Input & #3.\\
Output & A labeling coset $#2#4=\Lambda\leq\lab(V)$ such that:\\
(CL1) & $#2#4=\phi#2#5$ for all $\phi\in\iso(V;V')$.\\
(CL2) & $#2#4=#6\pi$#1for some (and thus for all) $\pi\in\Lambda$.
\end{tabular}
\end{prob}
}

\newcommand{\transitive}[2]{
Define $A^\ord:=A^\rho$ and define $\Delta^\ord:=(\Delta\rho)^\rho$.\\
Partition $A^\ord=A^\ord_1\cupdot A^\ord_2$ where
$A^\ord_1=\{1,\ldots,\lfloor\frac{|A^\ord|}{2}\rfloor\}$.\\
Compute the subgroup $\Psi^\ord:=\stab_{\Delta^\ord}(A^\ord_1)$.\\
\emph{(To compute the group $\Psi^\ord$, we use
a membership test.)}\\
Decompose $\Delta^\ord=\bigcup_{i\in[s]}\delta^\ord_i\Psi^\ord$ as union of left
cosets of the group $\Psi^\ord$.\\
Compute $\Delta_i\rho_i:=#1(#2,A,\rho\delta^\ord_i\Psi^\ord)$ for each
$i\in[s]$ recursively.\\
Rename the indices in $[s]$ such that:\\
$(#2,\Delta\rho)^{\rho_1}=\ldots=(#2,\Delta\rho)^{\rho_r} \prec(#2,\Delta\rho)^{\rho_{r+1}}\preceq\ldots\preceq
(#2,\Delta\rho)^{\rho_s}$.\\
Compute and return $\Lambda:=
\langle\Delta_1\rho_1,\ldots,\Delta_r\rho_r\rangle$.\\
\emph{(This is the smallest coset containing
$\Delta_1\rho_1\cup\ldots\cup\Delta_r\rho_r$.)} }

\newcommand{\minimal}{
\emph{(The minimality is w.r.t. the order ``$\prec$'' that is defined
in
\cref{sec:comb:objs:and:lab:cos}.)}}

\theoremstyle{plain}
\newtheorem{theo}{Theorem}
\newtheorem*{theo*}{Theorem}
\crefname{theo}{Theorem}{Theorems}
\newtheorem{lem}[theo]{Lemma}
\crefname{lem}{Lemma}{Lemmata}
\newtheorem{cor}[theo]{Corollary}
\crefname{cor}{Corollary}{Corollarys}

\theoremstyle{definition}
\newtheorem{defn}[theo]{Definition}

\crefname{exa}{Example}{Examples}
\newtheorem{prob}[theo]{Problem}
\crefname{prob}{Problem}{Problems}

\theoremstyle{remark}

\crefname{section}{Section}{Sections}
\crefname{appendix}{Appendix}{Appendices}
\crefname{figure}{Figure}{Figures}

\newenvironment{cl1}{\vspace{0.1cm}\noindent(\textit{CL1.})}{}
\newenvironment{cl2}{\vspace{0.1cm}\noindent(\textit{CL2.})}{}
\newenvironment{A}{\vspace{0.1cm}\noindent(\textit{A.})}{}
\newenvironment{runtime}{\vspace{0.1cm}\noindent(\textit{Running time.})}{}
\newcommand{\case}[1]{\item[\itshape\mdseries If {#1}:]}
\newenvironment{cs}{\begin{description}}{\end{description}}
\newcommand{\algorithmDAN}[1]{\vspace{0.1cm}\noindent
\underline{An algorithm for $#1$:}}

\definecolor{myBlue}{rgb}{0.0, 0.5, 1.0}
\definecolor{myRed}{rgb}{1.0, 0.5, 0}
\definecolor{myGreen}{rgb}{0.0, 0.5, 0.0}
\definecolor{myYellow}{rgb}{1.0, 1.0, 0.0}

\renewcommand{\phi}{\varphi}
\renewcommand{\epsilon}{\varepsilon}

\newcommand{\NN}{{\mathbb N}}

\newcommand{\ZZ}{{\mathbb Z}}

\newcommand{\CC}{{\mathcal C}}

\newcommand{\CO}{{\mathcal O}}

\newcommand{\CX}{{\mathcal X}}
\newcommand{\CY}{{\mathcal Y}}

\let\downarrowOld\downarrow
\renewcommand{\downarrow}{\mathord{\downarrowOld}}
\let\uparrowOld\uparrow
\renewcommand{\uparrow}{\mathord{\uparrowOld}}

\renewcommand{\tilde}{\widetilde}

\newcommand{\sym}{\operatorname{Sym}}

\newcommand{\aut}{\operatorname{Aut}}
\newcommand{\iso}{\operatorname{Iso}}
\newcommand{\stab}{\operatorname{Stab}}
\newcommand{\pstab}{\operatorname{PointStab}}
\newcommand{\norm}{\operatorname{Norm}}
\newcommand{\id}{\operatorname{id}}
\newcommand{\lab}{{\operatorname{Label}}}
\newcommand{\tranCl}{\operatorname{TClosure}}
\newcommand{\obj}{\operatorname{Objects}}

\newcommand{\set}{{\operatorname{Set}}}
\newcommand{\ord}{{\operatorname{Can}}}

\newcommand{\can}{\operatorname{CL}}
\newcommand{\canSet}{\operatorname{CL}_{\operatorname{Set}}}
\newcommand{\canObj}{\operatorname{CL}_{\operatorname{Object}}}
\newcommand{\canMatch}{\operatorname{CL}_{\operatorname{Match}}}
\newcommand{\canGroup}{\operatorname{CL}_{\operatorname{Group}}}
\newcommand{\canShift}{\operatorname{CL}_{\operatorname{Shift}}}

\begin{document}

\title{Normalizers and permutational
isomorphisms in
simply-exponential time}

\author{
Daniel Wiebking\\
RWTH Aachen University\\
\texttt{wiebking@informatik.rwth-aachen.de}}

\date{}

\maketitle

\KOMAoptions{abstract=true}

\begin{abstract}
We show that normalizers and permutational isomorphisms
of permutation groups given by generating sets can be computed
in time simply exponential in the degree of the groups.
The result is obtained by exploiting canonical forms
for permutation groups (up to permutational isomorphism).
\end{abstract}

\section{Introduction}

Computational group theory deals with practical computations in groups.
For example, today's computer algebra systems (such as Magma and
GAP) can handle permutation groups of degree in the hundreds of thousands.
The key for handling such large permutation groups lies in a compact implicit
representations of the groups.
In fact, every group $G\leq \sym(V)$
can be represented by a generating set of size polynomial in $|V|$, whereas
the order $|G|$ of a group can be exponential in $|V|$.
It is a priori unclear if fundamental algorithmic tasks as determining the order
of a group or testing membership in a group can be solved efficiently when
a group is given by such a compact representation.
In Sim's pioneering work, he gives
solutions for these tasks, which
have later been shown to run in polynomial time
\cite{SIMS1970169,Sims:1971:CPG:800204.806264,DBLP:conf/focs/FurstHL80}.
His algorithms are not only of theoretical interest, but randomized
versions as in \cite{DBLP:conf/issac/BabaiCFS91} turned out to be fast in
practice and build the core of present computer algebra systems (such as Magma and
GAP).

On the other side, there are central group theoretic tasks for which efficient
solutions are unavailable so far.
Important examples are the problems in the so-called Luks equivalent class,
which are all shown to be polynomial-time equivalent
\cite{DBLP:conf/dimacs/Luks91}.
The Luks class includes the setwise-stabilizer problem,
the group-intersection problem, the centralizer problem
and the string-isomorphism problem.
The class plays an important role even outside computational group theory.
In fact, Babai's recent breakthrough in graph isomorphism
is obtained by showing that the string-isomorphism
problem (and thus the entire Luks class) can be solved in quasipolynomial time \cite{DBLP:conf/stoc/Babai16}.

However, there is one important notoriously difficult problem outside the Luks
class that does not have a quasipolynomial-time solution yet.
This is the computation of the normalizer
$\norm(G)=\{h\in\sym(V)\mid h^{-1}Gh=G\}$ of a group $G\leq\sym(V)$.
This problem is also related to various isomorphism problems.
For example, the natural isomorphism problem for permutation groups, which asks
for a permutational isomorphism between two given groups,
reduces to the normalizer problem in polynomial time
\cite{DBLP:conf/dimacs/Luks91}.
Similarly, the isomorphism problem for linear codes, also known as code
equivalence, is polynomial time reducible to permutational isomorphism and thus
to the normalizer problem \cite{DBLP:conf/soda/BabaiCGQ11}.

Since the normalizer computation in general is a
tough challenge, researchers considered normalizers
$\norm(G)\cap H$ for restricted groups $H\leq\sym(V)$.
The first polynomial-time results were obtained for normalizers in
nilpotent groups, solvable groups, and $p$-groups
\cite{DBLP:conf/stoc/KantorL90,DBLP:conf/focs/Luks92,DBLP:conf/issac/LuksRW94}.
One decade later a result was obtained for normalizers $\norm(G)\cap H$ where both
groups $G$ and $H$ are restricted
\cite{DBLP:conf/issac/LuksM02}.
Another decade passed when these results were generalized to
normalizers where $H$ has restricted composition factors \cite{DBLP:journals/dmtcs/LuksM11}.

However, for normalizers without restrictions to the groups,
the best time bound (obtained via permutational isomorphism) is
$2^{\CO(|V|)}|G|^{\CO(1)}$ \cite{DBLP:conf/icalp/BabaiCQ12}.
Measured in terms of the degree of the group,
this is still no improvement over the brute force running time, since the order
$|G|$ of the group can be as large as the factorial of $|V|$.
It is a stated open problem whether the running time can be improved to $2^{\CO(|V|)}$
\cite{DBLP:conf/soda/BabaiCGQ11},\cite{codenotti2011testing}.

In this paper, we resolve this open problem and prove the following theorem.

\begin{theo*}

For groups $G\leq\sym(V)$ and $G'\leq\sym(V')$ given by generating sets, one can
compute the following tasks in time $2^{\CO(|V|)}$.
\begin{enumerate}
    \item Deciding permutational isomorphism, i.e., deciding whether
    there is a bijection $\phi:V\to V'$ such that $\phi^{-1}G\phi=G'$.
  \item Computing the normalizer $\norm(G)=\{h\in \sym(V)\mid h^{-1}Gh=G\}$.
  \item Computing a canonical labeling for $G$ (up to
  permutational isomorphisms).
\end{enumerate}

\end{theo*}

The first problem reduces to the second, whereas
the second task reduces to the third canonization problem.
On the other hand, no polynomial-time reduction from canonization to
isomorphism is known.
In fact, a lot research in group isomorphism has been done in recent
years and none of them seem to provide
canonical forms \cite{DBLP:conf/icalp/BabaiCQ12,DBLP:journals/gcc/LewisW12,
DBLP:journals/tcs/RosenbaumW15,
DBLP:conf/isaac/GrochowQ15,
DBLP:journals/siamcomp/GrochowQ17,
Brooksbank17}.

\paragraph*{Our Technique}
Recently, we introduced a
canonization framework to obtain canonical forms of various objects matching fastest known
isomorphism algorithms \cite{DBLP:journals/corr/abs-1806-07466}.
In this paper, we show that the canonization framework can be used
to obtain isomorphism and normalizer algorithms that are even faster than the
existing ones.
The approach of canonization instead of isomorphism allows the use of the object
replacement paradigm.
In particular, our main algorithm combines an object replacement lemma
with a decomposition technique of permutation groups into cosets to obtain
an efficient recursion.

To handle permutation groups also when they are given by generating sets, we
need to extend the recent canonization framework.
We therefore expand the notion of combinatorial objects by a new type of atom.
This allows the use of implicitly given permutation groups
and combines permutational group theory with powerful canonization techniques.

\paragraph*{Organization of the Paper}
In \cref{sec:comb:objs:and:lab:cos},
we extend the notion of
combinatorial objects in our recent framework.
The main difference is that groups occurring in combinatorial
objects are allowed to have implicit representations via generating sets.

In \cref{sec:groups}, the first algorithm shows how a canonical labeling for a
subgroup can be shifted to a canonical labeling for a coset.
The second algorithm finally gives a canonical labeling for
permutation groups in simply-exponential time and thereby proves our main theorem.

\section{Preliminaries}

\paragraph{Group Theory}
For $t\in\NN$, we define $[t]:=\{1,\ldots,t\}$.
The composition of two functions $f:V\to U$ and $g:U\to W$ is denoted by
$fg$ and is defined as the function that first applies $f$ and then applies $g$.
The symmetric group on a set $V$ is denoted by $\sym(V)$
and the symmetric group of degree $t\in\NN$ is denoted by $\sym(t)$.
A (permutation) group coset $X$ over a set $V$ is as set of permutations
such that $X=Hf=\{hf\mid h\in H\}$ for some subgroup $H\leq\sym(V)$
and some permutation $f\in\sym(V)$.
Analogous to subgroups, we say that $Hf$ is a \emph{subcoset} of a
coset $Gf$, written $Hf\leq Gf$, if $Hf$ is a subset of $Gf$ that again forms a coset.
In the following let $G\leq\sym(V)$ be a group.
The normalizer of $Hf\leq\sym(V)$ in $G$ is denoted by
$\norm_G(Hf):=\{g\in G\mid g^{-1}Hfg=Hf
\text{ for all }g\in G\}$.
The setwise
stabilizer of $A\subseteq V$ in $G$ is denoted by
$\stab_G(A):=\{g\in G\mid g(a)\in A
\text{ for all }a\in A\}$.
The pointwise stabilizer of $A\subseteq V$ in $G$
is denoted by $\pstab_G(A):=\{g\in G\mid g(a)=a\text{ for all }a\in A\}$.
In the case that $G$ is the symmetric group on $V$, the subgroups
$\stab(A)_{\sym(V)}$, $\pstab(A)_{\sym(V)}$ and $\norm_{\sym(V)}(Hf)$ are also denoted as $\stab(A)$,
$\pstab(A)$ and $\norm(Hf)$, respectively.
A set $A\subseteq V$ is called $G$-invariant
if $\stab_G(A)=G$.
A partition $V=V_1\cupdot\ldots\cupdot V_t$ is called $G$-invariant
if each part $V_i$ is $G$-invariant.
The finest $G$-invariant partition (with non-empty parts) is called the $G$-orbit
partition. The parts of the orbit partition are also called $G$-orbits.
A group $G\leq\sym(V)$ is called transitive on a set $A\subseteq
V$ if $A$ is a $G$-orbit.

\paragraph{Generating Sets and Polynomial-Time Library}

For the basic theory of handling
permutation groups given by generating sets, we refer to \cite{seress}.
Indeed, most algorithms are based on strong generating sets.
However, given an arbitrary generating set, the Schreier-Sims algorithm is used to compute a
strong generating set (of size quadratic in the degree) in polynomial time.
In particular, we will use that the following tasks can be performed efficiently
when a group is given by a generating set.

\begin{enumerate}
  \item Given a vertex $v\in V$ and a
  group $G\leq\sym(V)$, the Schreier-Sims algorithm can be used to compute
  the pointwise stabilizer $\stab_G(v)$ in polynomial time.
  \item Given a group $G\leq\sym(V)$, a subgroup that has a polynomial time membership problem
  can be computed in time polynomial in the index and the degree of the subgroup.
\end{enumerate}

\section{Combinatorial Objects With Implicitly Given Group
Cosets}\label{sec:comb:objs:and:lab:cos}

We start to recall and extend our framework from
\cite{DBLP:journals/corr/abs-1806-07466}.
The crucial difference is that group cosets
occurring in combinatorial objects are allowed to have implicit representations
via generating sets.

\paragraph{Labeling Cosets}
A \emph{labeling coset} of a set $V$ 
is set of bijections $\Lambda$
such that
$\Lambda=\Delta\rho=\{\delta\rho\mid \delta\in\Delta\}$
for some subgroup $\Delta\leq\sym(V)$
and some bijection $\rho:V\to\{1,\ldots,|V|\}$.
We write $\lab(V)$ to denote
the labeling coset $\sym(V)\rho=
\{\sigma\rho\mid \sigma\in\sym(V)\}$
where $\rho:V\to\{1,\ldots,|V|\}$ is an arbitrary bijection.
Analogous to subgroups, a set $\Theta\tau$ is called a \emph{labeling subcoset}
of $\Delta\rho$, written $\Theta\tau\leq\Delta\rho$,
if the labeling coset $\Theta\tau$ is a subset of $\Delta\rho$.
The restriction of $\Delta\rho\leq\lab(V)$ to a
a $\Delta$-invariant set $A\subseteq V$ is defined as $(\Delta\rho)|_A:=\{\lambda|_A\mid \lambda\in\Delta\rho\}$.
Observe that the restriction $(\Delta\rho)|_A$ does not necessarily form a
labeling coset since
the $\rho(A)$ might be a set
of natural numbers that differs from $\{1,\ldots,|A|\}$.
To rectify this, let $\kappa$ be the unique bijection from
$\rho(A)$ to $\{1,\ldots,|A|\}$ that preserves the
ordering ``$<$'' on the natural numbers.
The \emph{induced labeling coset} of $\Delta\rho$ on $A\subseteq V$ is defined as the labeling coset
$(\Delta\rho) \downarrow_A:=(\Delta\rho)|_A\kappa$.

\paragraph{Hereditarily Finite Sets and Combinatorial Objects}
In contrast to the previous framework, we will model group cosets
as a third kind of atom in order to represent them implicitly.

Inductively, we define \emph{hereditarily finite sets}
over a ground set $V$.
Each vertex $X\in V$ and each
labeling coset $Y=\Delta\rho\leq\lab(V)$
and also each group coset $Z=Gf\leq\sym(V)$ is called an \emph{atom}.
Each atom is in particular a hereditarily finite set.
Inductively, if $X_1,\ldots,X_t$ are hereditarily finite sets,
then also $\CX=\{X_1,\ldots,X_t\}$ and
$\CX=(X_1,\ldots,X_t)$ are hereditarily finite sets where $t\in\NN\cup\{0\}.$
A \emph{(combinatorial) object} is a pair $(V,\CX)$ consisting of a ground set $V$ and a
hereditarily finite set $\CX$ over $V$.
The ground set $V$ is usually apparent
from context and
the combinatorial object $(V,\CX)$ is identified with
the hereditarily finite set $\CX$.
The set $\obj(V)$ denotes the set of all (combinatorial)
objects over $V$. 
An object is called \emph{ordered} if the ground set $V$ is linearly
ordered.
The linearly ordered ground sets that we consider are always
subsets of natural numbers with their standard ordering ``$<$''.
An object is \emph{unordered}
if $V$ is a usual set (without a given order).
Partially ordered objects in which in which some, but not all, atoms are
comparable are not considered.

\paragraph{Representation of Objects With Implicitly Given Group Cosets}
All labeling cosets but also all group cosets occurring
as atoms will be represented concisely by generating sets.
This is the precise reason why we model them as an atom
rather than a hereditary finite set.

For an object $\CX\in\obj(V)$,
the \emph{transitive closure} of $\CX$, denoted by $\tranCl(\CX)$,
is defined
as all objects that recursively occur in $\CX$, i.e.,
$\tranCl(X):=\{X\}$ for $X=v\in V$
or $X=\Delta\rho\leq\lab(V)$ or $X=Gf\leq\sym(V)$.
Inductively, the transitive closure is defined as
$\tranCl(\CX):=\{\CX\}\cup \bigcup_{i\in[t]} \tranCl (X_i)$
for $\CX=\{X_1,\ldots,X_t\}$ or
$\CX=(X_1,\ldots,X_t)$.
All objects are efficiently represented
as colored directed acyclic graphs over
the elements in its transitive closure.
Using this representation,
the input size (with implicitly given group cosets)
of an object $\CX$ is polynomial in $|\tranCl(\CX)|+|V|+t_{\operatorname{max}}$
where $t_{\operatorname{max}}$ is the maximal length of a tuple in
$\tranCl(\CX)$.

\paragraph{Applying Functions to Unordered Objects}
Let $V$ be an unordered ground set and let $V'$
be a ground set that is either ordered or unordered.
The image of an unordered object $\CX\in\obj(V)$
under a bijection $\mu:V\to V'$
is an object $\CX^\mu\in\obj(V')$ that is defined as follows.
The image $X^\mu$ of a vertex $X=v\in V$ is defined as $\mu(v)$.
The image $Y^\mu$ of a labeling coset $Y=\Delta\rho\leq\lab(V)$
is defined as $\mu^{-1}\Delta\rho$.
The image $Z^\mu$ of a group coset $Z=Gf\leq\sym(V)$
is defined as $\mu^{-1}Gf\mu$.
Inductively, 
the image $\CX^\mu$ of an object $\CX=\{X_1,\ldots,X_t\}$ is defined
$\{X_1^\mu,\ldots,X_t^\mu\}$.
Similarly, the image $\CX^\mu$ of an object $\CX=
(X_1,\ldots,X_t)$ is defined as
$(X_1^\mu,\ldots,X_t^\mu)$.
Notice that the way we apply functions to cosets
does not depend on whether they are modeled as
an atom or as a hereditarily finite set.

\paragraph{Isomorphisms and Automorphisms of Unordered Objects}
The set of all isomorphisms
from an object $\CX\in\obj(V)$ and to an object $\CX'\in\obj(V')$ is denoted by
$\iso(\CX;\CX'):=\{\phi:V\to V'\mid \CX^\phi=\CX'\}$.
The set of all automorphisms of an object $\CX$, denoted by
$\aut(\CX):=\iso(\CX;\CX)$. Both isomorphisms and
automorphisms are defined for objects that are unordered only.
For specific objects, the automorphism group $\aut(\CX)$
often leads to a familiar notion,
e.g., $\aut((\Theta\tau,\Delta\rho))=\Theta\cap\Delta$,
$\aut((Gf,\Delta\rho))=\norm_\Delta(Gf)$
$\aut((A,\Delta\rho))=\stab_\Delta(A)$
and
$\aut((a_1,\ldots,a_t,\Delta\rho))=\pstab_\Delta(\{a_1,\ldots,a_t\})$
where $A\subseteq V$,
$\Theta\tau,\Delta\rho\leq\lab(V)$ and $Gf\leq\sym(V)$.

For two unordered sets $V$ and $V'$, the set $\iso(V;V')$ is also used to denote the set of
all bijections from $V$ to $V'$.
This notation indicates and stresses
that both $V$ and $V'$ have to be unordered.
Additionally, it is used in a context where
$\phi\in\iso(V;V')$ is seen as an isomorphism $\phi\in\iso(\CX;\CX^\phi)$.

\paragraph{The Linear Ordering of Ordered Objects}
We define a linear ordering of objects that remains polynomial-time computable
when group cosets are implicitly given.

For this paragraph, we
assume objects to be ordered where
$V\subseteq\NN$ and we
define a linear order ``$\prec$'' on them.
For two atoms $X,Y\in\NN$,
the natural ordering is adapted, i.e.,
$X\prec Y$ if $X<Y$.
For two sets $\CX=\{X_1,\ldots,X_s\}$ and
$\CY=\{Y_1,\ldots,Y_t\}$ on which the order ``$\prec$''
is already defined for the elements $X_i$ and $Y_j$, the
linear order is defined as follows.
We say that $\CX\prec\CY$ if
$|\CX|< |\CY|$ or if $|\CX|=|\CY|$ and the smallest
element in $\CX\setminus\CY$ is smaller
than the smallest element in $\CY\setminus\CX$.
For two tuples $\CX=(X_1,\ldots,X_s)$ and
$\CY=(Y_1,\ldots,Y_t)$
where ``$\prec$'' is already defined
for the entries, the
linear order is defined as follows. We say that $\CX\prec\CY$
if $s$ is smaller than $t$
or if $s=t$ and for the smallest position $i\in[t]$ for which
$X_i\neq Y_i$, it holds that
$X_i\prec Y_i$.
We extend the order to permutations
of natural numbers as follows.
For two permutations $\sigma_1,\sigma_2\in\sym(\{1,\ldots,|V|\})$
we say that $\sigma_1\prec \sigma_2$ if there
is an $i\in \{1,\ldots,|V|\}$ such that $\sigma_1(i) <\sigma_2(i)$
and $\sigma_1(j)=\sigma_2(j)$ for all $1\leq j<i$.
The definition is extended to
labeling cosets $\Delta\rho,\Theta\tau\leq\lab(\{1,\ldots,|V|\})$.
Similar to the case of sets, we say that
$\Delta\rho\prec\Theta\tau$ if
$|\Delta\rho|\leq|\Theta\tau|$ or if
$|\Delta\rho|=|\Theta\tau|$ and
the smallest element of $\Delta\rho\setminus \Theta\tau$
is smaller than the smallest element
of $\Theta\tau \setminus \Delta\rho$.
The definition for group cosets
$G_1f_1,G_2f_2\leq\sym(\{1,\ldots,|V|\})$ is analogous.
We say $G_1f_1\prec G_2f_2$ if
$|G_1f_1|\leq|G_2f_2|$ or if
$|G_1f_1|=|G_2f_2|$ and
the smallest element of $G_1f_1\setminus G_2f_2$
is smaller than the smallest element
of $G_2f_2 \setminus G_1f_1$.
Indeed, for two cosets $X,Y\leq\sym(\{1,\ldots,|V|\})=\lab(\{1,\ldots,|V|\})$
the ordering ``$\prec$'' for $X$ and $Y$ can be computed in time polynomial in
$|V|$ when the cosets are given by generating sets
(\cite{tree-width}, Corollary 22).
For completeness, we define
$X\prec Y\prec Z\prec \CX\prec \CY$ for all integers $X\in\NN$,
all labeling cosets $Y=\Delta\rho\leq\lab(\{1,\ldots,|V|\})$,
all group cosets $Z=Gf\leq\sym(\{1,\ldots,|V|\})$,
all tuples $\CX$ and all sets $\CY$. 
We say that $\CX\preceq\CY$ if $\CX=\CY$ or $\CX\prec\CY$.

The previous paragraph shows that the linear ordering remains polynomial-time
computable for objects that contain implicitly given group cosets.
\begin{lem}\label{lem:prec}
The ordering ``$\prec$'' on pairs of ordered objects can be
computed in polynomial time in the input size (with implicitly given group
cosets).
\end{lem}

We list the definitions and results obtained in the recent canonization
framework.

\begin{defn}[\cite{DBLP:journals/corr/abs-1806-07466}]
Let $\CC\subseteq\obj(V)$ be an isomorphisms-closed class of unordered objects.
A \emph{canonical labeling function} $\can$ is
a function that assigns each unordered object
$\CX\in\CC$ a labeling coset $\can(\CX)=\Lambda\leq\lab(V)$
such that:

\begin{enumerate}[(\textnormal{CL}1)]
  \item\label{ax:cl1} $\can(\CX)=\phi\can(\CX^\phi)$
  for all $\phi\in\iso(V;V')$ and,
  \item\label{ax:cl2} $\can(\CX)=\aut(\CX)\pi$
  for some (and thus for all) $\pi\in\can(\CX)$.
\end{enumerate}
In this case, the labeling coset $\Lambda$ is also called a \emph{canonical
labeling} for $\CX$.
\end{defn}

\begin{lem}[\cite{DBLP:journals/corr/abs-1806-07466}, Object
replacement lemma]\label{lem:rep} Let $\CX=\{X_1,\ldots,X_t\}$ be an object and
let $\can$ and $\canSet$ be canonical labeling functions.
Define $\CX^\set:=\{\Delta_1\rho_1,\ldots,\Delta_t\rho_t\}$
where $\Delta_i\rho_i:=\can(X_i)$ is a canonical labeling for $X_i\in\CX$.
Assume that
$X_i^{\rho_i}=X_j^{\rho_j}$ for all $i,j\in [t]$.
Then $\canObj(\CX):=\canSet(\CX^\set)$ defines a canonical labeling
for $\CX$.
\end{lem}

\begin{lem}[\cite{DBLP:journals/corr/abs-1806-07466}, Lemma
5]\label{lem:canMatch}
Canonical labelings for pairs $(M,\Delta\rho)\in\obj(V)$
where $M\subseteq V_1\times V_2$ is a matching
and $\Delta\rho$ is a labeling coset of $V=V_1\cupdot V_2$
can be computed in time
$2^{\CO(k_2)}|V|^{\CO(1)}$ where $k_2$ is the size of the largest $\Delta$-orbit of $V_2\subseteq V$.
\end{lem}

\begin{theo}[\cite{DBLP:journals/corr/abs-1806-07466},
Theorem 17]\label{theo:canObj} Canonical labelings for
objects $(\CX,\Delta\rho)\in\obj(V)$ can be computed in time
$2^{\CO(k)}n^{\CO(1)}$
where $n$ is the input size (when group cosets are explicitly given) and
$k$ is the size of the largest $\Delta$-orbit of $V$.
\end{theo}

\section{Canonization of Implicitly Given Permutation Groups}\label{sec:groups}

Before we canonize group cosets in general, we consider
canonical labeling functions for objects $\CX=(Gf,\Delta\rho)$
consisting of a group coset and a labeling coset
with a particular restriction to $G$ and $\Delta$.

\problem{\canShift\label{prob:CL:shift}}
{$(Gf,\Delta\rho)\in\obj(V)$ where $Gf\leq\sym(V)$,
$\Delta\rho\leq\lab(V)$, $V$ is an
unordered set and with the restriction that $\Delta=\norm_{\Delta}(G)$}
{(Gf,\Delta\rho)}
{(\phi^{-1}Gf\phi,\phi^{-1}\Delta\rho)}
{\norm_{\Delta}(Gf)}

The reader is encouraged to take a moment to verify that 
for input objects $\CX=(Gf,\Delta\rho)$, the
Conditions (CL1) and
(CL2) stated here agree with Conditions (CL1) and (CL2)
that are defined for objects in general.

The requirement that $\Delta=\norm_{\Delta}(G)$
says that the labeling coset $\Delta\rho$ consists of
canonical labelings for the group $G$.
Thus, the problem can be seen as the task of shifting the canonical labeling
$\Delta\rho=\norm_\Delta(G)\rho$ for the group $G$
to a canonical labeling $\Lambda=\norm_\Delta(Gf)\pi$ for the group coset $Gf$.

\begin{lem}\label{lem:can:shift}
A function $\canShift$
solving \cref{prob:CL:shift}
can be computed in time
$2^{\CO(k)}|V|^{\CO(1)}$
where
$k$ is the size of the largest $\Delta$-orbit of $V$.
\end{lem}

An isomorphism-version for groups with restricted composition factors
is stated as  Problem P7(1) in \cite{DBLP:conf/stoc/KantorL90}.
We use some of the ideas.
However, we make use of our framework to extend the algorithm to canonization
which later is required.
Additionally, we use the
adequate exponential recurrence that can handle unrestricted groups.

\begin{proof}
Define $\tilde V:=\{\tilde v_1,\ldots,\tilde v_{|V|}\}$ to
be a set of size $|V|$ that is disjoint from $V$.
The set $\tilde V$ can essentially be seen as a copy of the set $V$.
Define $U:=\tilde V\cupdot V$.
Define $\Delta_U\rho_U\leq\lab(U)$ to be the labeling coset
where $\Delta_U:=\{\delta_U\in\sym(U)\mid\exists
g\in G,\delta\in\Delta:\text{ for all }
i,j\in \{1,\ldots,|V|\} \text{ we have }
\delta_U(\tilde v_i)=\delta(g(v_i))$
and 
$\delta_U(v_j)=\delta(v_j)\}$
and where $\rho_U\in\lab(U)$
such that for all $i,j\in \{1,\ldots,|V|\}$ it holds
$\rho_U(\tilde v_i)=\rho(v_i)+|V|$ and $\rho_U(v_i)=\rho(v_i)$.
By the assumption that $\delta^{-1}G\delta=G$ for all $\delta\in\Delta$,
the set $\Delta_U$ defines indeed a group closed under composition.
Define a matching $M_f:=\{(\tilde{f(v_i)},v_i)\mid i\in \{1,\ldots,|V|\}\}$
by pairing the vertices according to $f$.
Compute $\Lambda_U:=\canMatch(M_f,\Delta_U\rho_U)$ using the algorithm from
\cref{lem:canMatch}.
We claim that the induced labeling coset $\Lambda:=\Lambda_U\downarrow_V$
defines
a canonical labeling for $(Gf,\Delta\rho)$
\emph{(the induced labeling coset
is defined at the beginning of \cref{sec:comb:objs:and:lab:cos})}.

\begin{cl1}
Assume we have $\phi^{-1}Gf\phi,\phi^{-1}\Delta\rho$
instead of $Gf,\Delta\rho$ as an input.
By the construction, we obtain $\phi_U^{-1}\Delta_U\rho_U$
instead of $\Delta_U\rho_U$ for some $\phi_U$ with $\phi_U|_V=\phi$.
Moreover, we obtain $M_f^{g_U\phi_U}$ instead of $M_f$
and for some $g_U$ with $g_U|_V=\id$
(the choices for $f\in Gf$ vary up to elements
in $G$).
By (CL1) of $\canMatch$ and since
$\phi_U^{-1}\Delta_U\rho_U=\phi_U^{-1}g_U^{-1}\Delta_U\rho_U$, we obtain
$\phi_U^{-1}g_U^{-1}\Lambda_U$ instead of $\Lambda_U$.
Therefore, we obtain
$(\phi_U^{-1}g_U^{-1}\Lambda_U)\downarrow_V=\phi^{-1}\Lambda$ instead of
$\Lambda$.
\end{cl1}

\begin{cl2}
In order to verify (CL2), we show that
$(\aut(M_f)\cap\Delta_U)|_V=\norm_\Delta(Gf)$.
The inclusion $\norm_\Delta(Gf)\subseteq(\aut(M)\cap\Delta_U)|_V$
already follows from Condition (CL1) of the problem.
It remains to show the reversed inclusion also holds, i.e.,
$(\aut(M_f)\cap\Delta_U)|_V\subseteq\norm_\Delta(Gf)$.
Let $\alpha:\tilde V\to V$ to be the bijection s
with $f(\tilde v_i)= v_i$.
So let $\delta_U\in\aut(M_f)\cap\Delta_U$.
Since $\delta_U\in\Delta_U$,
there are some $\delta\in\Delta,g\in G$ such that
$\alpha(\delta_U(\tilde v_i))=\delta(g(v_i))$
and $\delta_U(v_i)=\delta(v_i)$ for all $i\in \{1,\ldots,|V|\}$.
Since $\delta_U\in\aut(M_f)$, it holds that
$\alpha(\delta_U(\alpha^{-1} (v_i)))=f(\delta_U(f^{-1}(
v_i)))$ for all $i\in\{1,\ldots,|V|\}$.
Both together imply that $\delta(g(v_i))=\alpha(\delta_U(\tilde
v_i))=f(\delta_U(f^{-1}( v_i)))=f(\delta(f^{-1}(v_i)))$
for all $v_i\in V$.
Thus $\delta^{-1}f\delta=gf$ for some $g\in G$.
By assumption, it holds $\delta^{-1}G\delta=G$
and thus $\delta^{-1}Gf\delta=\delta^{-1}G\delta\delta^{-1}f\delta=Gf$.
This proves
$\delta_U|_{V}=\delta\in\norm_\Delta(Gf)$.
\end{cl2}

\begin{runtime}
All steps are polynomial-time computable, except the computation of
$\canMatch$.
The used algorithm runs in time simply exponential in $k_2$
where $k_2$ is
the size of the largest $\Delta_U$-orbit of $V\subseteq U$.
By the construction, $k_2$ is equal to
the size of the largest $\Delta$-orbit of $V$.
\end{runtime}
\end{proof}

Next, we will consider the general problem without that restriction to the
group and the labeling coset.

\problem{\canGroup\label{prob:CL:Group}}
{$(Gf,\Delta\rho)\in\obj(V)$ where $Gf\leq\sym(V)$,
$\Delta\rho\leq\lab(V)$ and $V$ is an unordered set}
{(Gf,\Delta\rho)}
{(\phi^{-1}Gf\phi,\phi^{-1}\Delta\rho)}
{\norm_\Delta(Gf)}

As usual, Conditions (CL1) and (CL2) given here coincide with the general
Condition (CL1) and (CL2).
For $\Delta=\sym(V)$ and $Gf=G$, the problem is equivalent to computing a
canonical labeling for a permutation group (up to permutational isomorphism).

\begin{theo}\label{theo:canGroup}
A function $\canGroup$
solving \cref{prob:CL:Group}
can be computed in time
$\binom{2|V|}{k}^{\CO(1)}\subseteq 2^{\CO(|V|)}$
where $k$ is the size of the largest
$\Delta$-orbit of $V$.
\end{theo}

\paragraph{Intuition for the Permutation Group Algorithm}
To solve \cref{prob:CL:Group}, we will maintain at any point in time
a set $A\subseteq V$ which is $\Delta$-invariant and $G$-invariant
and for which we require Condition (A):
$\Delta=\norm_{\Delta}(G_A)$ for $G_A:=\pstab_G(A)$.
Intuitively, this means 
that the labeling coset $\Delta\rho$ consists of canonical labelings for
the subgroup $G_A$ which is obtained from $G$ by a pointwise fixation of the set
$A$.
This set $A$ measures our progress in the sense
the index of $G_A$ in $G$ is bounded by $|A|!$
and thus if the set $A$ is small, then we have already canonized a
relatively large subgroup $G_A$ of $G$.

In the transitive case in which the set $A$ is a $\Delta$-orbit,
we decompose the labeling coset $\Delta\rho$ into (intransitive) labeling
subcosets.
Each labeling subcoset can be seen as an individualization of the permutation
domain and can be handled recursively.
Each recursive call leads to the case of intransitive labeling cosets which we
explain next.

In the intransitive case where $\Delta\leq\stab(A_1,A_2)$ and where
$G\leq\stab(A)$, we define a subgroup chain $G_A=G_{A_1\cup A_2}\leq G_{A_1}\leq
G_{\{A_1\}}\leq G$.
The subgroups are defined as $G_{A_1}:=\pstab_G(A_1)$ and $G_{\{A_1\}}:=\stab_G(A_1)$.
In a bottom up fashion, we will compute canonical labelings for all these
groups.
By Condition (A), a canonical labeling of the group $G_A=G_{A_1\cup A_2}$ is
already available.
Since $G_{A_1\cup A_2}$ is a subgroup of $G_{A_1}$ obtained by a
pointwise fixation of the set $A_2$, a canonical labeling for $G_{A_1}$ can be obtained
recursively using the same algorithm with $G_{A_1}$ in the role of $G$ and $A_2$
in the role of $A$.
After this recursive call, the canonical labeling for $G_{A_1}$ is available.
Similar, the group $G_{A_1}$ is a subgroup of $G_{\{A_1\}}$ obtained by
pointwise fixation of the set $A_1$, so also a canonical labeling for
$G_{\{A_1\}}$ can be obtained recursively.
To compute a canonical labeling for $G$, we use a different strategy that
exploits that the index of $G_{\{A_1\}}$ in $G$ is simply-exponentially bounded.
We decompose $G=\bigcup_{g\in G} G_{\{A_1\}}g$ into right cosets of $G_{\{A_1\}}$.
Using the previous shifting algorithm, we can shift the canonical
labeling for $G_{\{A_1\}}$ (which is already computed) to all the cosets
$G_{\{A_1\}}g$.
So far, the algorithm computed canonical labelings for all the cosets $G_{\{A_1\}}g$.
Finally, we make use of an object replacement paradigm together with the main
algorithm of our recent framework to combine the collection of the canonical labelings.
By doing so, this results in a canonical labeling for the entire group
$G$.

\paragraph{Detailed Description of the Permutation Group Algorithm} 
Proving \cref{theo:canGroup}, we give a detailed description and analysis
of the algorithm for permutation groups and cosets.

\begin{proof}[Proof of \cref{theo:canGroup}]
For the purpose of recursion, we need an additional input parameter.
Specifically, we use a subset $A\subseteq V$ such that $G,\Delta\leq\stab(A)$
and such that
\begin{enumerate}
  \item[\textnormal{(A)}] $\Delta=\norm_{\Delta}(G_A)$ for $G_A:=\pstab_G(A)$.
\end{enumerate}
Initially, we set $A:=V$.

\algorithmDAN{\canGroup(Gf,A,\Delta\rho)}
\begin{cs}

\case{$Gf\neq G$}~\\
Compute $\Lambda_1:=\canGroup(G,A,\Delta\rho)$ recursively.\\
Compute and return $\Lambda:=\canShift(Gf,\Lambda_1)$ using the algorithm from
\cref{lem:can:shift}.

\end{cs}
\emph{(Now, we achieved that $Gf=G\leq\sym(V)$ is a group, rather than a
proper coset.)}
\begin{cs}
\case{$|A|\leq 1$}~\\
Return $\Lambda:=\Delta\rho$.

\case{$\Delta$ is intransitive on $A$}~\\
Partition $A=A_1\cupdot A_2$ where
$A_1\subseteq A$ is a minimal size
$\Delta$-orbit such that $A_1^\rho$ is minimal.\\
\minimal\\
Define $G_{A_1}:=\pstab_G(A_1)$ and define $G_{\{A_1\}}:=\stab_G(A_1)$.\\
\emph{(To compute the group $G_{A_1}$, we use the
Schreier-Sims algorithm and to compute the group $G_{\{A_1\}}$, we use
a membership test. The running times are specified in the preliminaries.)}\\
Compute $\Lambda_1:=\canGroup(G_{A_1},A_2,\Delta\rho)$ recursively.\\
Compute $\Lambda_2:=\canGroup(G_{\{A_1\}},A_1,\Lambda_1)$ recursively.\\
Define $\CX:=\{G_{\{A_1\}}g\mid g\in G\}$.\\
Compute $\Delta_{X}\rho_{X}:=\canShift(X,\Lambda_2)$ for all $X=G_{\{A_1\}}g\in \CX$
using the algorithm from \cref{lem:can:shift}.\\
Define $\CX^\set:=\{\Delta_{X}\rho_{X}\mid X=G_{\{A_1\}}g\in \CX\}$.\\
Define an \emph{ordered} partition
$\CX^\set=\CX_1^\set\cupdot\ldots\cupdot \CX_s^\set$ 
such that\\
$X^{\rho_X}
\prec Y^{\rho_Y}$, if and only if
$\Delta_{X}\rho_{X}\in \CX^\set_p$
and $\Delta_{Y}\rho_{Y}\in \CX^\set_q$ for some~$p,q\in[s]$ with~$ p<q$.\\
Return
$\Lambda:=\canObj((\CX_1,\ldots,\CX_s),\Delta\rho)$ using the algorithm from
\cref{theo:canObj}.

\case{$\Delta$ is transitive on $A$}~\\
\transitive{\canGroup}{G}

\end{cs}

\begin{A}
In the intransitive case, we need to show that Condition (A) remains
satisfied for the recursive calls.
Consider the first recursive call.
By definition, Condition (A) remains satisfied for the recursive instance
$(G_{A_1},A_2,\Delta\rho)$ also, i.e., $\Delta=\norm_\Delta(G_A)
=\norm_\Delta((G_{A_1})_{A_2})$.
For the second recursive call observe that
$\Lambda_1=\Delta_1\rho_1$ is a canonical labeling for
$(G_{A_1},A_2,\Delta\rho)$ and thus by (CL2) of this problem, it holds
$\Delta_1=\norm_{\Delta_1}(G_{A_1})=\norm_{\Delta_1}((G_{\{A_1\}})_{A_1})$.
\end{A}

\begin{cl1}
Assume that we have $\phi^{-1}Gf\phi,A^\phi,\phi^{-1}\Delta\rho$
instead of $M,A,\Delta\rho$ as an input.
For these parameters $\phi^{-1}\rho$ is a coset representative
for $\phi^{-1}\Delta\rho$.
We need to show that the algorithm outputs $\phi^{-1}\Lambda$
instead of $\Lambda$.

In the case $Gf\neq G$, we obtain $\phi^{-1}\Lambda_1$
instead of $\Lambda_1$ by induction.
By (CL1) of $\canShift$, we return $\phi^{-1}\Lambda$
instead of $\Lambda$.

In the base case $|A|\leq 1$, we
return $\phi^{-1}\Delta\rho$ instead of $\Delta\rho$.

In the intransitive case, we obtain
the same partition $A^\phi=A_1^\phi\cupdot A_2^\phi$
since $A_1^\rho=A_1^{\phi\phi^{-1}\rho}$.
We obtain $(\phi^{-1}G\phi)_{A_1^\phi}=\phi^{-1}G_{A_1}\phi$
and $(\phi^{-1}G\phi)_{\{A_1^\phi\}}=\phi^{-1}G_{\{A_1\}}\phi$
instead of $G_{A_1}$ and $G_{\{A_1\}}$, respectively.
By induction, we obtain $\phi^{-1}\Lambda_1$
and $\phi^{-1}\Lambda_2$ instead of $\Lambda_1$
and $\Lambda_2$, respectively.
We obtain $\CX^\phi$ instead of $\CX$
and by (CL1) of $\canShift$, we return $\phi^{-1}\Delta_X\rho_X$
instead of $\Delta_X\rho_X$.
By (CL1) of $\canObj$, we finally return $\phi^{-1}\Lambda$
instead of $\Lambda$.

The transitive case is similar to the analysis used to proof
\cref{lem:canMatch}, however for completeness, we will recall it.
The ordered objects $A^\ord$ and $\Delta^\ord$
remain unchanged since
$A^{\phi\phi^{-1}\rho}=A^{\rho}$
and $(\phi^{-1}\Delta\rho)^{\phi^{-1}\rho}=(\Delta\rho)^{\rho}$.
Also the partition $A^\ord=A^\ord_1\cupdot A^\ord_2$
and the ordered group $\Psi^\ord$ remains unchanged.
We obtain cosets of the form $\phi^{-1}\rho\delta_i^\ord\Psi^\ord$
instead of $\rho\delta^\ord_j\Psi^\ord$ since
the indexing is arbitrary.
The calls are
of the form
$\canGroup(\phi^{-1}G\phi,A^\phi,\phi^{-1}\delta^\ord_i\Psi^\ord)$
instead of $\canGroup(G,A,\delta_j^\ord\Psi^\ord)$.
By induction, we obtain $\phi^{-1}\Delta_i\rho_i$
instead of $\Delta_j\rho_j$.
Therefore, we obtain $\phi^{-1}\rho_i$
instead of $\rho_j$.
However, the ordered sequence remains unchanged since $(G,\Delta\rho)^{\rho_i}
=(\phi^{-1}G\phi,\phi^{-1}\Delta\rho)^{\phi^{-1}\rho_i}$.
The computation of $\Lambda$ is known to be isomorphism invariant
and therefore the algorithm returns $\phi^{-1}\Lambda$
instead of $\Lambda$.
\end{cl1}

\begin{cl2}
In the base case where $|A|\leq 1$, we have that $G_A=\stab_G(A)=G$.
Combined with Condition (A), it follows that
$\Delta=\norm_\Delta(G_A)=\norm_\Delta(G)$.

In the intransitive case, it holds that $\Lambda$
defines a canonical labeling for $(\CX_1^\set,\ldots,\CX_s^\set,\Delta\rho)$.
By object replacement (\cref{lem:rep}),
$\Lambda$ defines
a canonical labeling for $(\CX_1,\ldots,\CX_s,\Delta\rho)$
where $\CX=\CX_1\cupdot\ldots\cupdot \CX_s$
such that
$X\in \CX_i$, if and only if $\Delta_X\rho_X\in \CX^\set_i$.
Since $(\CX_1,\ldots,\CX_r)$ is an \emph{ordered}
partition of $\CX$ defined in an isomorphism-invariant way, it holds that
$\Lambda$ defines a canonical labeling for $(\CX,\Delta\rho)$.
Again, $\CX=\{G_{\{A_1\}}g\mid g\in G\}$ is an isomorphism-invariant
(unordered) partition of the group $G$
and therefore $\Lambda$ defines a canonical labeling for $(G,\Delta\rho)$.

Consider the transitive case and
observe that Condition (CL1)
of the problem $\canGroup$ already
implies that
$\norm_\Delta(G)\pi\subseteq\Lambda$
for some $\pi\in\Lambda$.
It remains to show that the reversed inclusion also holds, i.e.,
$\Lambda\subseteq\norm_\Delta(G)\pi$.
Equivalently, we need to show  $\rho_i\rho_j^{-1}\in\norm_\Delta(G)$ for all
$i,j\in[r]$.
The membership $\rho_i\rho_j^{-1}\in\norm(G)$
follows from the equation $G^{\rho_i}=G^{\rho_j}$
and the membership $\rho_i\rho_j^{-1}\in\Delta$ follows similarly
from the equation $(\Delta\rho)^{\rho_i}=(\Delta\rho)^{\rho_j}$.
\end{cl2}

\begin{runtime}
First, we consider the number of recursive calls of this algorithm.
In Case $Gf\neq G$, there is one single recursive call
that leading to a different case,
so this case can be neglected in our analysis.
It remains to consider the recursive
calls of the other cases.
Let $A^*\subseteq A$ be a $\Delta$-orbit that is of maximal size.
We claim that the maximum number of recursive calls $R(|A^*|,|A|)$ is bounded
by $T:=2^{6|A^*|}|A|^2$.
In the intransitive case, this can be seen by induction:
\begin{equation*}
R(|A^*|,|A|)\leq 1+\sum_{j\in[2]} R(|A^*|,|A_j|)
\overset{\text{induction}}{\leq}
1+2^{6|A^*|}(|A_1|^2+|A_2|^2)\leq T.
\end{equation*}
In the transitive case, it holds that $A^*=A$
and $s\leq 2^{|A|}$ and which leads to
\begin{equation*}
R(|A|,|A|)\leq 1+ s\cdot R(\lceil|A|/2\rceil,|A|)
\overset{\text{induction}}{\leq}
1+2^{4|A|+3}|A|^2
\overset{2\leq |A|}{\leq}
T.
\end{equation*}

Next, we give a bound on the running time that is needed for one single call
without recursive costs.
Let $k$ be the size of the
largest $\Delta$-orbit for the initial instance.
In Case $Gf\neq G$, we use the algorithm from \cref{lem:can:shift}
which runs in time $2^{\CO(k)}|V|^{\CO(1)}$.
In the intransitive case, we need to compute a group $G_{\{A_1\}}$.
The index of $G_{\{A_1\}}$ in $G$ is bounded by
the $G$-orbit of $A_1$
which in turn is at most
$b:=\binom{|V|}{|A_1|}\leq\binom{2|V|}{k}$.
The group $G_{\{A_1\}}$ can be computed
in time polynomial in the index and $|V|$,
i.e., $(b|V|)^{\CO(1)}$.
We have $b$ calls to the algorithm from \cref{lem:can:shift}
that
run in time $2^{\CO(k)}|V|^{\CO(1)}$ per instance.
Similar, the representation size of $\CX^\set$ is bounded by $(b|V|)^{\CO(1)}$
and therefore $\Lambda$ can be computed in time 
$2^{\CO(k)}(b|V|)^{\CO(1)}$.
In the transitive case, the group $\Psi^\NN$ is computed
in time polynomial in the index and $|V|$,
i.e., ${2^{\CO(k)}|V|^{\CO(1)}}$.

In total, we have a running time of at most
$T\cdot 2^{\CO(k)}(b|V|)^{\CO(1)}
\subseteq \binom{2|V|}{k}^{\CO(1)}$.
\end{runtime}
\end{proof}

\begin{cor}\label{cor:canGroup}
Canonical labelings for permutation groups and cosets (up to permutational
isomorphism) can be computed in time $\binom{2|V|}{k}^{\CO(1)}\subseteq 2^{\CO(|V|)}$
where $V$ is the permutation domain and
$k$ is the size of the largest color class of $V$.
\end{cor}

\begin{proof}
We need to a compute canonical labelings
for $(Gf,C_1,\ldots,C_t)$
where $Gf\leq\sym(V)$
is a group coset over the domain
$V=C_1\cupdot\ldots\cupdot C_t$.
This is done by calling the previous algorithm with input $(Gf,\Delta\rho)$
where $\Delta\rho=\{\lambda\in\lab(V)\mid\forall i,j\in[t],i<j \forall v_i\in
C_i,v_j\in C_j:\lambda(v_i)<\lambda(v_j)\}$.
\end{proof}

\section{Outlook and Open Questions}

We showed that canonization of permutation groups on a permutation domain $V$
can be done in $2^{\CO(|V|)}$ regardless of the order $|G|$ of the group.
Our result generalizes the known time bound of
$2^{\CO(|V|)}|G|^{\CO(1)}$.

However, in the setting of bounded color
class size $k$, we did not achieved a generalization.
In the this work, we presented an algorithm
for implicitly given permutation groups
running in time
$|V|^{\CO(k)}$.
However, for explicitly given groups $G$, an algorithm is known that
runs in time $2^{\CO(k)}|V|^{\CO(1)}|G|^{\CO(1)}$
\cite{DBLP:journals/corr/abs-1806-07466}.
Since the running times are orthogonal to each other,
we ask for an unifying running time for canonization of permutation groups in $2^{\CO(k)}|V|^{\CO(1)}$.

With our result for permutational isomorphisms
most of the studied isomorphism problems have simply-exponential time
bounds now.
However, the isomorphism problem for implicitly given group codes (also known as
group code equivalence) does not have such a time bound.
While the problem for linear codes reduces to permutational
isomorphism in polynomial time, the situation for cyclic groups seems difficult.
It is an open problem if group code equivalence for cyclic groups $G$ of prime
power order (such as $\ZZ/p^2\ZZ$) and codes of length $|V|$ can be decided in
$2^{\CO(|V|)}\log(|G|)^{\CO(1)}$.

\bibliographystyle{alpha}
\bibliography{references}

\addcontentsline{toc}{section}{Bibliography}

\end{document}